\documentclass[11pt]{article}

\usepackage[T1]{fontenc}
\usepackage{microtype}
\usepackage{amsmath,amsthm}
\usepackage{newpxtext,newpxmath}
\usepackage[margin=15pt,font=small,labelfont={bf}]{caption}
\usepackage[margin=1in]{geometry}
\usepackage[english]{babel}
\usepackage[utf8]{inputenc}

\usepackage[backend=biber,style=alphabetic,maxalphanames=4,minalphanames=4,maxbibnames=99,minbibnames=99]{biblatex}
\usepackage{graphicx}
\usepackage{algorithm}
\usepackage[noend]{algpseudocode}

\AtBeginDocument{
  }
\usepackage{xcolor}
\usepackage[colorlinks=true,
    citecolor=darkgreen,
    linkcolor=darkred
]{hyperref}
\usepackage{aliascnt}
\definecolor{darkgreen}{rgb}{0,0.5,0.5}
\definecolor{darkred}{rgb}{0.5,0,0}
\newcommand\drop[1]{}
\newcommand*{\newaliastheorem}[3]{
  \newaliascnt{#1}{#2}
  \newtheorem{#1}[#1]{#3}
  \aliascntresetthe{#1}
  \expandafter\newcommand\csname #1autorefname\endcsname{#3}
}
\newtheorem{theorem}{Theorem}
\newaliastheorem{proposition}{theorem}{Proposition}
\newaliastheorem{lemma}{theorem}{Lemma}
\newaliastheorem{corollary}{theorem}{Corollary}
\newaliastheorem{definition}{theorem}{Definition}
\newaliastheorem{conjecture}{theorem}{Conjecture}
\newaliastheorem{example}{theorem}{Example}
\newaliastheorem{remark}{theorem}{Remark}

\addto\extrasenglish{

}

\title{Online Coloring for Graphs of Large Odd Girth}
\author{
Hirotaka Yoneda\thanks{Supported by JSPS JP25K24465, JST ASPIRE JPMJAP2302, and JST ACT-X JPMJAX25CT.} \\
The University of Tokyo, Japan \\
squar37@gmail.com
\and
Masataka Yoneda\thanks{Supported by JSPS JP25K24465 and JST ASPIRE JPMJAP2302.} \\
The University of Tokyo, Japan \\
e869120@gmail.com
}
\date{}
\begin{document}
\maketitle

\pagenumbering{gobble}

\begin{abstract}
We study the problem of online coloring for graphs with large odd girth. The best previously known algorithm uses $O(n^{1/2})$ colors, which was discovered by Kierstead in 1998. This algorithm works when the odd girth is 7 or more. In this paper, we provide the following: 
for every $\varepsilon > 0$, there exists a constant $g' \in \{3, 5, 7, \dots\}$ such that graphs with odd girth at least $g'$ can be deterministically colored online using $O(n^{\varepsilon})$ colors. 
\end{abstract}

\pagenumbering{arabic}
\setcounter{page}{1}

\section{Introduction}

We study the \emph{online coloring} problem, which concerns online algorithms for graph coloring. In this setting, the vertices are revealed one by one, and we must assign each vertex a color immediately after it is revealed, so that its color differs from the colors of all its neighbors. The goal is to design a strategy that minimizes the number of colors used.

However, since the algorithm has no knowledge of future vertices, it requires a clever strategy to color the graph effectively. For example, one strategy is to assign the least-indexed available color to each arriving vertex. This strategy is called First-Fit \cite{GL88}. However, it costs $\frac{1}{2} n$ colors even for bipartite graphs (\autoref{fig:intro-ff}), which shows the ineffectiveness of First-Fit.

The problem is known to be difficult for general graphs, where a deterministic algorithm that achieves a competitive ratio of $O(\frac{n}{\log \log n})$ was very recently given by Kawarabayashi, Yoneda, Yoneda \cite{KYY25}. In contrast, the competitive ratio of any algorithm is $\Omega(\frac{n}{(\log n)^2})$, due to Halld\'{o}rsson and Szegedy \cite{HS94}.

Given strong negative results for general graphs, effective online coloring algorithms for special graph classes have also received attention. For example, for tree, bipartite, planar, bounded-treewidth, or disk graphs, deterministic algorithms that use $O(\log n)$ colors are known \cite{Bea76, LST89, Ira94, AS21}. For interval graphs, a deterministic algorithm with a competitive ratio of $3$ is given by Kierstead and Trotter \cite{KT81}.

\begin{figure}[b]
    \centering
    \includegraphics[width=\linewidth]{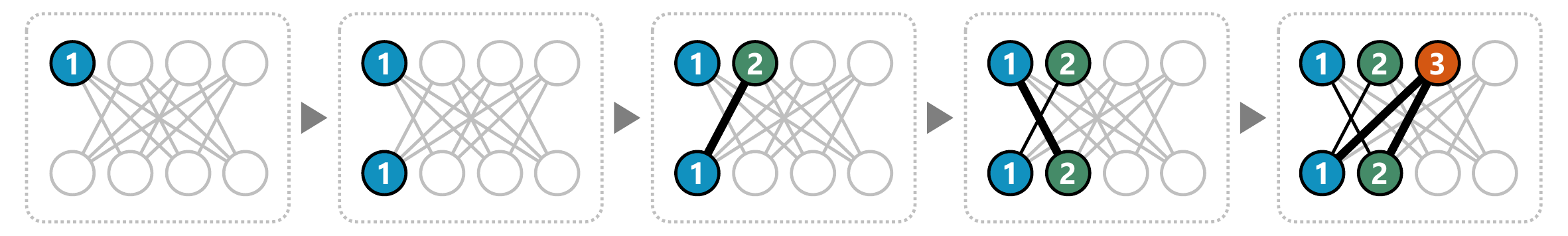}
    \caption{A worst-case input for the First-Fit algorithm, which uses $\frac{1}{2} n$ colors for bipartite graphs, where $n$ is the number of vertices. The number written inside each vertex is the color.}
    \label{fig:intro-ff}
\end{figure}

\subsection{The Odd Girth: Almost Bipartite Graphs}

In graph theory, graphs with large odd girth have been widely studied. The odd girth $g$ is the length of the shortest odd cycle in a graph. This is closely connected to the chromatic number --- as suggested by the intuition that graphs with large odd girth can be interpreted as ``almost bipartite graphs'' (cf. bipartite graphs have no odd cycle).

The chromatic number of graphs with odd girth at least $5$ (i.e., triangle-free graphs) is a classical problem in graph theory, which is also studied as the asymptotics of the Ramsey number $R(3, k)$. Mycielski (1955) \cite{mycie} introduced a construction for such graphs with arbitrarily large chromatic number. Later, the maximum chromatic number of such graphs was shown to be $\Theta(\sqrt{n / \log n})$, where the lower bound is given by Ajtai, Koml\'{o}s, Szemer\'{e}di (1980) \cite{AKS80} and the upper bound is given by Kim (1995) \cite{Kim95}. On the other hand, the existence of graphs with an arbitrarily large girth (not just ``odd girth'') and an arbitrarily large chromatic number was shown by the famous result of Erd\H{o}s (1959) \cite{Erdos1959}. This was a landmark result that established the proof technique of probabilistic method.

To the best of our knowledge, the best-known bounds on the chromatic number of graphs with general constant odd girth $g \geq 5$ are $O((\frac{n}{\log n})^{\frac{2}{g-1}})$ by Denley (1994) \cite{Den94}, and the existence of graphs with chromatic number $\Omega(n^{\frac{1}{g-3}} / \log n)$ was given by Krivelevich (1995) \cite[Theorem 1]{Kri95}.

\subsection{Previous Results}

Graphs with odd girth $g$ had been extensively studied also in online coloring.

\begin{description}
    \item[\textbf{Algorithms.}] For the case $g = 5$ (triangle-free graphs), the best-known deterministic algorithm by Brada\v{c} et al. (2026) \cite{BFS+26} uses $O(\frac{n \log \log n}{\log n})$ colors. For the case $g = 7$ ($(C_3, C_5)$-free graphs), Kierstead (1998) \cite{Kie98} presented a deterministic algorithm that uses $O(\sqrt{n})$ colors. For general odd girth $g \geq 9$, Nagy-Gy\"{o}rgy (2010) \cite{Nag10} presented a deterministic algorithm that uses $O(\sqrt{\frac{n \log g}{g}})$ colors. However, the $O(\sqrt{n})$ bound still stands when $g$ is a constant.
    \item[\textbf{Lower bounds.}] It is known that any deterministic algorithm uses $\Omega(n^{1/2})$ colors for the case $g = 5$ \cite{DKV13} and $\Omega((\frac{n}{\log n})^{1/3})$ colors for the case $g = 7$ \cite{GKM+14}. For $g \geq 9$, the chromatic number lower bound by Krivelevich \cite{Kri95} remains the best known.
\end{description}

A large gap remains between the upper and lower bounds. A natural question is whether the 28-year-old square-root barrier can be broken and $O(n^{1/2-\varepsilon})$ colors can be achieved for some $\varepsilon > 0$. If so, another question is how much further the exponent can be improved.

\subsection{Our Contributions}

In this paper, we answer the 28-year-old question affirmatively with the following theorem:

\begin{theorem}\label{thm:generalize-1}
    For every $\varepsilon > 0$, there exists a constant $g' \in \{3, 5, 7, \dots\}$ such that graphs with odd girth at least $g'$ can be deterministically colored online using $O(n^{\varepsilon})$ colors.
\end{theorem}

More precisely, we present a deterministic online algorithm that colors graphs with odd girth $g \geq \frac{11}{2} \cdot 5^k + \frac{3}{2}$ using $O(k \cdot n^{\frac{2}{k+4}})$ colors, for every integer $k \geq 0$ (\autoref{thm:generalize-2}). This significantly improves the previous result of $O(n^{1/2})$ colors by Kierstead (1998) \cite{Kie98}.

Our result also achieves $O(\log n)$ colors for graphs with odd girth $g = \Omega(n^c)$, where $c \ (0 < c < 1)$ is a constant (\autoref{cor:generalize-3}). This substantially improves the previous result of $O(n^{(1-c)/2} \sqrt{\log n})$ colors by Nagy-Gy\"{o}rgy (2010) \cite{Nag10} and matches the $\Theta(\log n)$-color bound for bipartite graphs \cite{LST89}. \autoref{tab:colors-girth} shows the relationship between the odd girth and the number of colors used.

\begin{table}[t]
    \centering
    \caption{Comparison of our algorithm with previous results by odd girth. $c \ (0 < c < 1)$ is a constant.}
    \begin{tabular}{|l|rrrrrrr|}
        Layers $k$ & $0$ & $1$ & $2$ & $3$ & $\cdots$ & $c \log_5 n$ \\ \hline
        Odd Girth & $\geq 7$ & $\geq 29$ & $\geq 139$ & $\geq 689$ & $\cdots$ & $\geq \Omega(n^c)$ \\
        Previous \cite{Kie98, Nag10} & $O(n^{0.500})$ & $O(n^{0.500})$ & $O(n^{0.500})$ & $O(n^{0.500})$ & $\cdots$ & $O(n^{(1-c)/2} \sqrt{\log n})$ \\
        Our Result & $O(n^{0.500})$ & $O(n^{0.400})$ & $O(n^{0.334})$ & $O(n^{0.286})$ & $\cdots$ & $O(\log n)$ \\
    \end{tabular}
    \label{tab:colors-girth}
\end{table}

For online coloring in general, the bottleneck to achieving $O(n^{1/2-\varepsilon})$ colors in many algorithms, including Kierstead's $O(\sqrt{n})$-color algorithm for graphs of odd girth $g \geq 7$ and Vishwanathan's $O(\sqrt{n})$-color \emph{randomized} algorithm for $3$-colorable graphs \cite{Vis92}, is that they only utilize the structure of the (first) neighborhood.

Our algorithm makes use of the second neighborhood, and as a result, we achieve $O(n^{2/5})$ colors for graphs of odd girth $g \geq 29$ (\autoref{thm:main-1}). In addition, we generalize this result to the third neighborhood and further neighbors, to obtain our main result (\autoref{thm:generalize-1}). The challenge is that when each vertex has degree at least $\Delta$, the size of the $k$-th neighborhood may not be $\Omega(\Delta^k)$; in the worst case, the size can be as small as $O(\Delta)$. We overcome this issue by using a novel framework called \emph{online group coloring}, which allows us to apply an approach similar to First-Fit and benefit in cases where these neighborhoods are small.

\subsection{Odd Girth vs. Girth}

Some readers may wonder how many colors are needed for online coloring of graphs with girth $g$ (instead of odd girth). In this case, First-Fit performs effectively, where the following result was shown by Zaker (2007) \cite{Zak07}. We provide a simpler proof here:

\begin{proposition}[\cite{Zak07}] \label{prop:girth-ff}
    Let $k \geq 1$ be an integer. For a graph $G = (V, E)$ of girth $g \geq 2k+1$, First-Fit uses only $O(k \cdot n^{1/k})$ colors.
\end{proposition}

\begin{figure}[h]
    \centering
    \includegraphics[width=\linewidth]{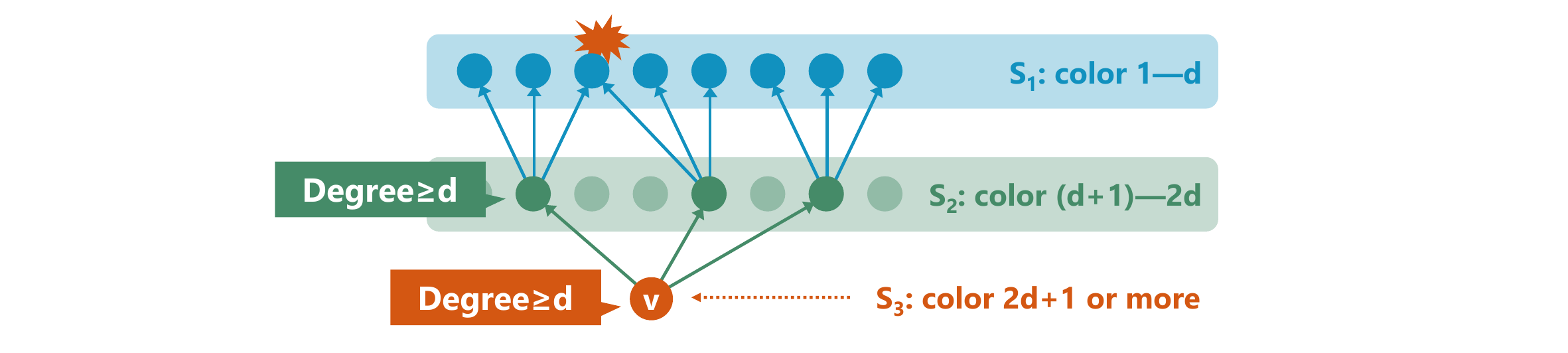}
    \caption{A sketch of our analysis for the case $k = 2$. If the endpoints of two paths clash, a cycle of length at most $4$ is formed, a contradiction. Hence, there must be at least $d^2$ vertices in $S_1$.}
    \label{fig:intro-girth}
\end{figure}

\begin{proof}
    Let $d = \lceil n^{1/k} \rceil$. For each $i = 1, \dots, k$, let $S_i$ be the set of vertices that are colored by colors $(i-1)d + 1, \dots, id$. Also, let $S_{k+1} = V \setminus (S_1 \cup \dots \cup S_k)$. Then, for each $i = 2, \dots, k+1$ and $v \in S_i$, we have $|N_{S_{i-1}}(v)| \geq d$. This is because $v$ must be adjacent to some vertices with each of color $(i-2)d + 1, \dots, (i-1)d$ by the definition of First-Fit.
    
    Suppose there exists a vertex $v$ in $S_{k+1}$. Then, there are at least $d^k$ paths of length $k$ that start from $v$ and pass through a vertex in $S_k, \dots, S_1$ in this order. For all of these paths, the endpoints (in $S_1$) are distinct. Otherwise, in the two such paths with a common endpoint, a cycle of length at most $2k$ would appear, which contradicts the girth requirement. Therefore, $|S_1| \geq d^k \ (\geq n)$ must hold (\autoref{fig:intro-girth}). However, we also have $|S_1| \leq n-1$, a contradiction. Thus, $S_{k+1}$ must be empty, meaning that First-Fit uses at most $kd = O(k \cdot n^{1/k})$ colors.
\end{proof}

We note that the above result matches the best-known upper bound $O((\frac{n}{\log n})^{1/k})$ for chromatic number of graphs with girth $g \geq 2k+1$, given by Denley \cite{Den94}, up to poly-logarithmic factors. Also, one can derive from the standard Moore bound argument \cite{HS60, God13} that any graph of girth $g \geq 2k+1$ is an $n^{1/k}$-degenerate graph. Hence, the result by Irani \cite{Ira94} that First-Fit only uses $O(d \log n)$ colors for $d$-degenerate graphs implies the bound of $O(n^{1/k} \log n)$ colors. \autoref{prop:girth-ff} is better by a factor of $\log n$.

In contrast, for graphs with large odd girth, First-Fit may perform much worse, since it can use $\Theta(n)$ colors even for bipartite graphs (see \autoref{fig:intro-ff}). This is a fundamental difference between ``girth'' and ``odd girth'' in the context of online coloring.

\subsection{Organization}

\autoref{sec:prelim} presents the preliminaries needed for the proof. \autoref{sec:group} presents the online group coloring problem and its solution, which is a key tool in our algorithm. \autoref{sec:special} presents an algorithm that uses $O(n^{2/5})$ colors for graphs with odd girth $g \geq 29$. \autoref{sec:general} extends the idea to multiple layers and obtains our main result, \autoref{thm:generalize-1}.

\section{Preliminaries} \label{sec:prelim}

Before moving on to our algorithm, we present the notation and definitions, along with Kierstead's $O(\sqrt{n})$-color algorithm for $(C_3, C_5)$-free graphs.

\subsection{Notation and Definitions}

First, we formally define the online coloring problem.

\begin{definition}[\cite{KYY25}]
    In the online coloring problem, an undirected graph $G = (V, E)$ is initially empty; in other words, there are zero vertices. Vertices are added to $G$ one by one, together with their incident edges. Immediately after a vertex $v$ is added (along with its incident edges), we must assign a color $c(v) \in \mathbb{N}$ to $v$. This color must differ from that of all its neighbors. The next vertex is presented only after assigning a color. The goal is to minimize the number of colors used, that is, $|\{c(v) : v \in V\}|$.
\end{definition}

For notation, given a graph $G = (V, E)$, we denote the set of vertices adjacent to $v \in V$ as $N(v)$ or $N_G(v)$, and the set of vertices adjacent to at least one vertex in $S \subseteq V$ as $N(S)$ or $N_G(S)$. Note that when describing online coloring algorithms, notation such as $N(v)$ and $N(S)$ is taken from the current graph (not the final graph). We also define the \emph{odd girth} along with its useful property.

\begin{definition}
    The odd girth of a graph $G = (V, E)$ is the length of the shortest odd-length cycle in $G$, or $+\infty$ if no such cycle exists. The odd girth is also equal to the length of the shortest odd-length closed walk in $G$.
\end{definition}

In the following, we define the notion of \emph{even-diameter} in a graph, a concept that is used extensively throughout our online coloring algorithms.

\begin{definition}
    Let $G = (V, E)$ be a graph. The ``even-distance'' between two vertices $s, t \in V$ is the length of the shortest even-length walk between $s$ and $t$, or $+\infty$ if no such walk exists. The ``even-diameter'' of a vertex set $S \subseteq V$ is the maximum even-distance between any two vertices in $S$. Note that the walks are allowed to pass through vertices outside $S$.
\end{definition}

\subsection{Kierstead's Algorithm} \label{subsec:kierstead}

Next, we review Kierstead's algorithm for graphs without $3$-cycles and $5$-cycles. First, it is a standard assumption in online coloring that the final number of vertices $n$ is known in advance \cite{Kie05, KYY25}. This is justified by the following lemma.

\begin{lemma}[\cite{KYY25}] \label{lem:n-is-known}
    Let $f: \mathbb{N} \to \mathbb{R}_{\geq 0}$ be a non-decreasing unbounded function. For a specific graph class (e.g., graphs with odd girth at least 7), suppose there exists a deterministic online coloring algorithm for $n$-vertex graphs that uses at most $f(n)$ colors, for every $n \in \mathbb{N}$. Then there exists a deterministic online coloring algorithm for this graph class that uses at most $4f(n)$ colors, even when the final number of vertices, $n$, is unknown.
\end{lemma}

\begin{theorem}[\cite{Kie98}] \label{thm:kierstead}
    For graphs of odd girth at least 7, there is a deterministic online coloring algorithm that uses $O(n^{1/2})$ colors.
\end{theorem}

\begin{proof}
    The idea of his algorithm is to divide vertices into ``good'' sets $X_1, X_2, \dots \ (\subseteq V)$, where for each $i$, the vertices of $X_i$ are adjacent to a common vertex $w_i \in V$. The convenience of these sets is that a group of vertices adjacent to at least one vertex in $X_i$, say $Y_i \ (\subseteq V)$, can be colored with the same color (we color them with color $\lceil n^{1/2} \rceil + i$). Indeed, if two such vertices $y_1, y_2 \in Y_i$ are adjacent, this forms a closed walk of length $5$: $y_1 \to x_1 \to w_i \to x_2 \to y_2 \to y_1$, where $x_j$ is a vertex in $X_i$ adjacent to $y_j$ (for $j = 1, 2$). This contradicts the assumption of odd girth $\geq 7$.

    Next, we describe the algorithm. Let $X_1, \dots, X_r$ be the current list of these ``good'' sets. Initially, $r = 0$. When a new vertex $v \in V$ arrives, we perform the following procedure:
    
    \begin{enumerate}
        \item If $v$ can be colored by any of colors $1, \dots, \lceil n^{1/2} \rceil$, color $v$ using First-Fit (\autoref{fig:kierstead}a).
        \item Otherwise, if there exists some $i$ such that $X_i \cap N(v) \neq \emptyset$, we pick one such $i$ and color $v$ with color $\lceil n^{1/2} \rceil + i$ (\autoref{fig:kierstead}b). This also means that $v$ is added to group $Y_i$ (i.e., $Y_i \gets Y_i \cup \{v\}$).
        \item Otherwise, we create a new ``good'' set $X_{r+1} := N(w_{r+1})$ with $w_{r+1} := v$ and then we color $v$ with color $\lceil n^{1/2} \rceil + (r+1)$  (\autoref{fig:kierstead}c). This also means that $v$ is added to group $Y_{r+1}$ (i.e., $Y_{r+1} \gets Y_{r+1} \cup \{v\}$). This increases $r$ by $1$.
    \end{enumerate}
    
    For each $i$, we have $|X_i| \geq \lceil n^{1/2} \rceil$ since $w_i$ is not colored in Step 1 (First-Fit). Also, $X_j \cap X_i = \emptyset$ for each $j = 1, \dots, i-1$ since $w_i$ is not colored in Step 2. Therefore, $r \leq n^{1/2}$; we use only $\lceil n^{1/2} \rceil + r = O(n^{1/2})$ colors in total.
\end{proof}

\begin{figure}[t]
  \centering
  \includegraphics[width=\linewidth]{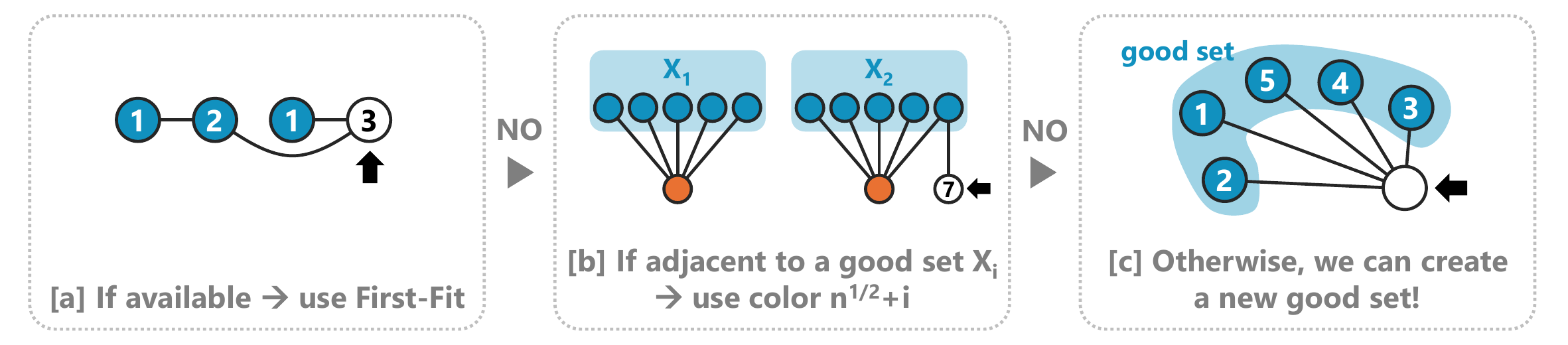}
  \caption{A sketch of Kierstead's algorithm when $\lceil n^{1/2} \rceil = 5$. The white vertex with an arrow is a newly arrived vertex $v$ that is to be colored. The numbers inside the vertices are the assigned colors.}
  \label{fig:kierstead}
\end{figure}

\section{Online Group Coloring} \label{sec:group}

In this section, we define a novel problem, \emph{online group coloring}, which is a key tool in our algorithm. This is a variant of online coloring where vertices are divided into groups.

\begin{definition}
    The online group coloring problem is a special case of online coloring where each vertex $v$ belongs to a group $g(v) \in \mathbb{N}$. The group to which $v$ belongs is revealed at the time of its arrival. As in online coloring, we need to color $v$ immediately. It is guaranteed that there are no edges between vertices in the same group.

    This problem specifies a parameter $\Delta \geq 0$. The ``neighbor graph'' $H$ is defined as follows: each group forms a vertex, and there is an edge between two groups if and only if at least one edge exists between them in $G$. When a vertex $v$ arrives, the degree of $g(v)$ in $H$ is guaranteed to be at most $\Delta$ (it is allowed for the degrees of other vertices to exceed $\Delta$ as a result).
\end{definition}

\begin{theorem}\label{thm:group-coloring-1}
    There exists a deterministic algorithm that uses at most $\Delta^2 + 2$ colors for online group coloring.
\end{theorem}

\begin{proof}
    The strategy is to try to keep using the same color inside each group, but if this color is already used in adjacent groups (in the neighbor graph $H$), we switch to a new color. For each group $i$, let $c_i$ be the current color for this group, and let $L_i$ be the set of colors ever used for this group. For convenience, we initially set $c_i = 1$ and $L_i = \{1\}$.

    When a vertex $v \in V$ arrives, we normally color $v$ with color $c_{g(v)}$. However, this color may already be used in adjacent groups, i.e., when $c_{g(v)} \in \bigcup_{i \in N_H(g(v))} L_i$. In this case, we update $c_{g(v)}$ to the smallest color not in $\bigcup_{i \in N_H(g(v))} L_i$, and also add this new color to $L_{g(v)}$. It is easy to see that the same color is not used for adjacent vertices.
    
    Next, we analyze the number of colors used. First, each $L_i$ is updated only if, compared to the time when the previous vertex in group $i$ arrived, either of the following holds: (1) a group newly becomes adjacent to group $i$ (in $H$), or (2) $L_j$ for an already adjacent group $j$ (in $H$) newly starts to include $c_i$. However, case (2) means that $c_j$ was updated to the same color as $c_i$ after that time, but $c_i$ was already contained in $L_i$, so this cannot occur. Also, $L_i$ will not be updated after the final vertex in group $i$ is processed. Hence, $L_i$ is updated at most $\Delta$ times, so $|L_i \setminus \{1\}| \leq \Delta$. Thus, for each arrival of $v \in V$:
    \begin{equation*}
        c_{g(v)} \leq \left|\bigcup_{i \in N_H(g(v))} L_i\right| + 1 = \left|\bigcup_{i \in N_H(g(v))} (L_i \setminus \{1\})\right| + 2 \leq \sum_{i \in N_H(g(v))} |L_i \setminus \{1\}| + 2
    \end{equation*}
    which is at most $\Delta^2 + 2$ by the degree constraint $|N_H(g(v))| \leq \Delta$.
\end{proof}

\section{The Algorithm with $O(n^{2/5})$ colors} \label{sec:special}

In this section, we present an algorithm that uses $O(n^{2/5})$ colors for graphs with odd girth $\geq 29$. This improves the previous $O(n^{1/2})$-color bound due to Kierstead \cite{Kie98}.

\subsection{Overview} \label{subsec:special-overview}

In \autoref{subsec:kierstead}, Kierstead's algorithm covered $V$ with ``good'' sets $X_1, \dots, X_r$, where the even-diameter of each $X_i$ is only $2$. His observation was that with odd girth $\geq 7$, the entire group $Y_i$ can be colored using the same color.

In this section, as we consider graphs with odd girth $\geq 29$, the requirements for ``good'' sets are relaxed: their even-diameters only need to be at most $24$. Hence, our goal is to cover $V$ with $O(n^{2/5})$ large sets of even-diameter $\leq 24$. Below, we sketch the intuition for our algorithm to achieve this goal. 

\begin{itemize}
    \item We attempt to directly improve Kierstead's algorithm to $O(n^{2/5})$ colors. Since only $n^{2/5}$ colors can be used for First-Fit, $r = O(n^{3/5})$ sets $X_1, \dots, X_r$ of even-diameter $2$ are created, each having size $\Omega(n^{2/5})$ (\autoref{fig:ours-sketch}a). However, since we need to color $Y_1, \dots, Y_r$ within $O(n^{2/5})$ colors, we must use the same color across multiple groups $Y_i$.
    \item The key observation is that when a group $Y_i$ is adjacent to at least $\Omega(n^{1/5})$ other groups $Y_j$, we can ``merge'' these neighboring even-diameter-$2$ sets $X_j$ to form a larger set of even-diameter $\leq 24$, of size $\Omega(n^{3/5})$ (\autoref{fig:ours-sketch}b). In this way, we can cover $V$ with $O(n^{2/5})$ sets of even-diameter $\leq 24$.
    \item On the other hand, when each group has only $O(n^{1/5})$ adjacent groups, the coloring procedure reduces to the \emph{online group coloring} problem with parameter $\Delta := \Theta(n^{1/5})$. This enables us to use the same color across different groups  (\autoref{fig:ours-sketch}c).
\end{itemize}

\begin{figure}[t]
  \centering
  \includegraphics[width=\linewidth]{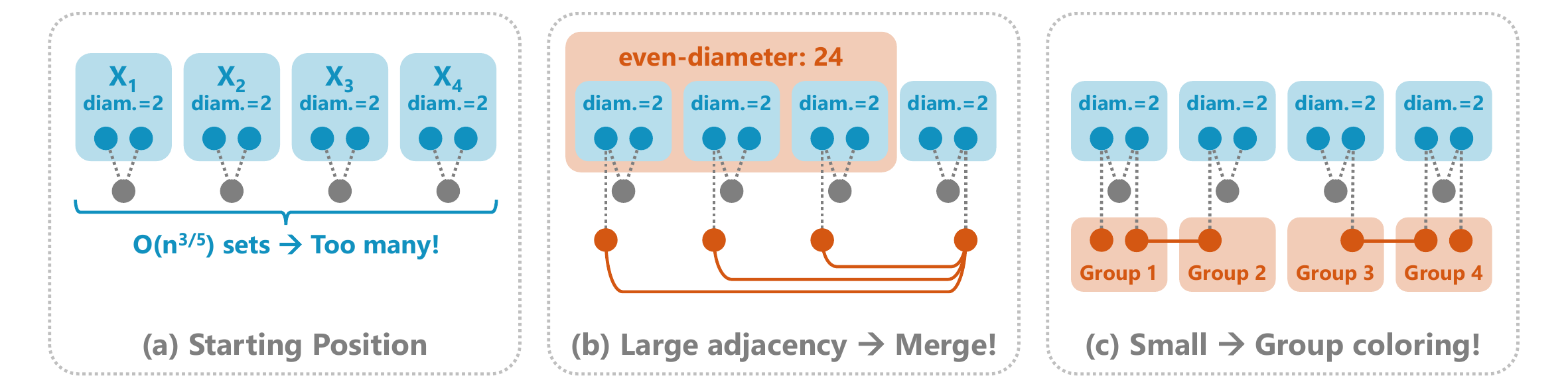}
  \caption{A sketch of our method. (a) In Kierstead's algorithm where we use $\lceil n^{2/5} \rceil$ colors for  First-Fit, we need to handle $O(n^{3/5})$ sets. (b) When a group is adjacent to many groups, we merge these groups to form a large set with even-diameter $\leq 24$. (c) When a group is adjacent to few groups, we use the procedure of online group coloring.}
  \label{fig:ours-sketch}
\end{figure}

\subsection{Formalization}

As a foundation for building our algorithm, we formalize the procedure we must perform for coloring non-First-Fit vertices in Kierstead's algorithm in a generalized manner. We define the subproblem ``$(r^*, d^*)$-subroutine'' as follows. Note that in Kierstead's algorithm, the procedure for coloring non-First-Fit vertices is equivalent to the $(n^{1/2}, 2)$-subroutine.

\begin{definition}
    We consider the $(r^*, d^*)$-subroutine, where $r^* \in \mathbb{R}_{\geq 0}$ and $d^* \in \{2, 4, 6, \dots\}$ are parameters. In this subroutine, we maintain sets $X_1, \dots, X_r \subseteq V$ called ``bases'' and sets $Y_1, \dots, Y_r \subseteq V$ called ``groups''. Initially, there are no such sets, i.e., $r = 0$. The following types of events may occur in an arbitrary order (see also \autoref{fig:subroutine}):

    \begin{description}
        \item[\textbf{Base addition.}] A new base $X_{r+1} \subseteq V$ is added. It is guaranteed that the even-diameter of $X_{r+1}$ is at most $d^*$ (in the graph $G$ for online coloring). The corresponding group $Y_{r+1}$ is initially empty. This increases $r$ by $1$.
        \item[\textbf{Coloring query.}] A vertex $v \in V$ is added to a group $Y_i$. It is guaranteed that $v$ is adjacent to at least one vertex in $X_i$. Then, we need to color $v$ so that its color differs from the colors of all adjacent vertices in $Y_1 \cup \dots \cup Y_r$. (We ignore the coloring on $V \setminus (Y_1 \cup \cdots \cup Y_r)$ because we assume that their colors are disjoint from the colors in $Y_1 \cup \cdots \cup Y_r$.)
    \end{description}

    It is guaranteed that at most $r^*$ bases are added, i.e., $r \leq r^*$. The objective is to minimize the number of colors used across all coloring queries. Note that the bases $X_1, \dots, X_r$ are not assumed to be disjoint, but the groups $Y_1, \dots, Y_r$ are clearly disjoint by definition.
\end{definition}

\begin{figure}[t]
  \centering
  \includegraphics[width=\linewidth]{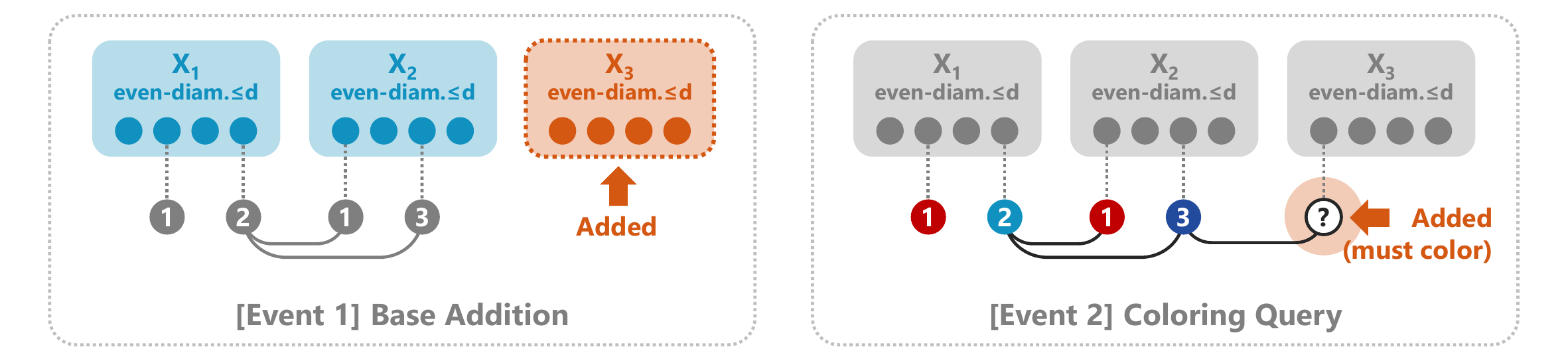}
  \caption{Two types of events in the $(r^*, d^*)$-subroutine. The number of bases added is guaranteed to be at most $r^*$. Note that the vertices in $X_i$ can be regarded as ``uncolored'' because we assume that their colors are disjoint from the colors we use in $Y_1 \cup \cdots \cup Y_r$.}
  \label{fig:subroutine}
\end{figure}

\subsection{Dividing the Problem}

We solve the online coloring problem for graphs with odd girth $\geq 29$ using the following strategy. Part A and Part C can be solved similarly to Kierstead's algorithm. The key difficulty lies in Part B, where we use online group coloring and merge bases of even-diameter $2$ to obtain a large set with even-diameter $\leq 24$.

\begin{description}
    \item[\textbf{Part A.}] Reduce online coloring to the $(n^{3/5}, 2)$-subroutine using $O(n^{2/5})$ extra colors.
    \item[\textbf{Part B.}] Reduce the $(n^{3/5}, 2)$-subroutine to the $(n^{2/5}, 24)$-subroutine using $O(n^{2/5})$ extra colors.
    \item[\textbf{Part C.}] Solve the $(n^{2/5}, 24)$-subroutine using $O(n^{2/5})$ colors.
\end{description}

First, we solve Part A. Intuitively, if we only use $\lceil n^{2/5} \rceil$ colors for First-Fit in Kierstead's algorithm, the problem of coloring the remaining vertices falls to the $(n^{3/5}, 2)$-subroutine. Below, we formally prove this fact.

\begin{lemma} \label{lem:special-base-case}
    For graphs with odd girth at least 29, if there exists a deterministic algorithm $\mathcal{A}$ that solves the $(n^{3/5}, 2)$-subroutine with $O(n^{2/5})$ colors, then there exists a deterministic online coloring algorithm with $O(n^{2/5})$ colors.
\end{lemma}

\begin{proof}
    We use an algorithm very similar to \autoref{thm:kierstead}. To solve online coloring, we use the $(n^{3/5}, 2)$-subroutine, denoted by $\mathcal{S}_0$. We use disjoint color palettes for $\mathcal{S}_0$ and for First-Fit (i.e., we use colors $\lceil n^{2/5} \rceil + 1$ and above in $\mathcal{S}_0$). After that, whenever a vertex $v \in V$ arrives, we perform the following procedure:

    \begin{enumerate}
        \item If $v$ can be colored by any of colors $1, \dots, \lceil n^{2/5} \rceil$, we color $v$ using First-Fit.
        \item Otherwise, if $v$ is adjacent to any existing base in $\mathcal{S}_0$, we color $v$ inside $\mathcal{S}_0$. Formally, we pick one such base and add $v$ to the corresponding group for this base in $\mathcal{S}_0$ (as a coloring query). This means that we use the algorithm $\mathcal{A}$ to color it.
        \item Otherwise, we add $N(v)$ as a base for $\mathcal{S}_0$. After that, we can color $v$ by going back to Step 2, because $v$ is adjacent to a new base $N(v)$.
    \end{enumerate}

    The pseudocode for this algorithm is given in \autoref{alg:reduction}. The input to the $(n^{3/5}, 2)$-subroutine is valid, because (1) each base added to $\mathcal{S}_0$ clearly has even-diameter $\leq 2$, and (2) the number of bases created is at most $\frac{n}{\lceil n^{2/5} \rceil} \leq n^{3/5}$, shown by the same argument as \autoref{thm:kierstead}. For number of colors used, we use $\lceil n^{2/5} \rceil$ colors for First-Fit, and another $O(n^{2/5})$ colors in $\mathcal{S}_0$ (using the algorithm $\mathcal{A}$). Thus, we use $O(n^{2/5})$ colors in total.
\end{proof}

\begin{algorithm}[t]
    \caption{A reduction from online coloring to the $(n^{3/5}, 2)$-subroutine, using $O(n^{2/5})$ colors.}
    \label{alg:reduction}
    \begin{algorithmic}[1]
        \State $\mathcal{S}_0 \gets$ the $(n^{3/5}, 2)$-subroutine
        \State $V_{\mathrm{FF}} \gets \emptyset$
        \For {each arrival of $v \in V$}
            \If {$v$ can be colored by any of colors $1, \dots, \lceil n^{2/5} \rceil$}
                \State $V_{\mathrm{FF}} \gets V_{\mathrm{FF}} \cup \{v\}$, and color $v$ with First-Fit
            \ElsIf {$v$ is adjacent to any existing base in $\mathcal{S}_0$}
                \State color $v$ inside $\mathcal{S}_0$ (by adding $v$ to $\mathcal{S}_0$ as a coloring query)
            \Else
                \State add $N(v)$ as a base for $\mathcal{S}_0$
                \State go to lines 6--7 to color $v$ (because $v$ neighbors a new base $N(v)$)
            \EndIf
        \EndFor
    \end{algorithmic}
\end{algorithm}

Next, we solve Part C. Intuitively, when the odd girth is $\geq 29$, we can use the same color for each group as in Step 2 of Kierstead's algorithm. Below, we formally prove this fact.

\begin{lemma} \label{lem:special-part-b}
    For graphs with odd girth at least 29, there exists a deterministic algorithm that solves the $(n^{2/5}, 24)$-subroutine with $n^{2/5}$ colors.
\end{lemma}

\begin{proof}
    Given that each base $X_i$ has even-diameter $\leq 24$, no two vertices in group $Y_i$ are adjacent. Suppose that $y_1, y_2 \in Y_i$ are adjacent. Let $x_j$ be a vertex in $X_i$ that is adjacent to $y_j$ (for $j = 1, 2$). Then there exists a closed walk $y_1 \to x_1 \to \cdots \to x_2 \to y_2 \to y_1$, where $x_1 \to \cdots \to x_2$ is an even-length walk of length at most $d^* = 24$. This is an odd closed walk of length at most $d^* + 3 = 27$, which contradicts the assumption of odd girth $\geq 29$.
    
    Therefore, when a vertex $v$ is put into group $Y_i$ in a coloring query, we can color $v$ with color $i$. We use at most $r^* = n^{2/5}$ colors in total.
\end{proof}

\subsection{The Main Algorithm}

It remains to solve Part B, that is, to solve the $(n^{3/5}, 2)$-subroutine. We solve this part using the $(n^{2/5}, 24)$-subroutine, which is denoted as $\mathcal{S}_1$, and the online group coloring problem where ``group $i$'' is used for a subset of vertices in $Y_i$. Note that we use disjoint color palettes for $\mathcal{S}_1$ and for online group coloring. To outline the strategy, when a vertex $v$ arrives in a coloring query, we color it in the following way:

\begin{description}
    \item[\textbf{Case 1.}] We color $v$ inside the $(n^{2/5}, 24)$-subroutine $\mathcal{S}_1$ whenever it is possible.
    \item[\textbf{Case 2.}] Otherwise, we color $v$ in either of the following ways: (a) by putting it into online group coloring, or (b) by creating new bases for the $(n^{2/5}, 24)$-subroutine.
\end{description}

We maintain, for each group $Y_i$, the set of vertices $Y'_i \subseteq Y_i$ that proceed to ``Case 2'' (which is Step 2 in the following algorithm). We also maintain a graph $H^+ = (\{1, \dots, r\}, E(H^+))$ which represents the current adjacency between these subsets $Y'_i$. The set of edges is $E(H^+) := \{i-j:$ there exist some edges between $Y'_i$ and $Y'_j$ in the original graph $G$$\}$. It is clear that the neighbor graph for online group coloring is a subgraph of $H^+$.

Now, we describe the algorithm. The parameter for online group coloring is set to $\Delta := 6n^{1/5}$. For each coloring query, supposing that vertex $v$ arrives in group $Y_i$, we perform the following procedure  (see also \autoref{fig:main-algorithm}):

\begin{figure}[t]
  \centering
  \includegraphics[width=\linewidth]{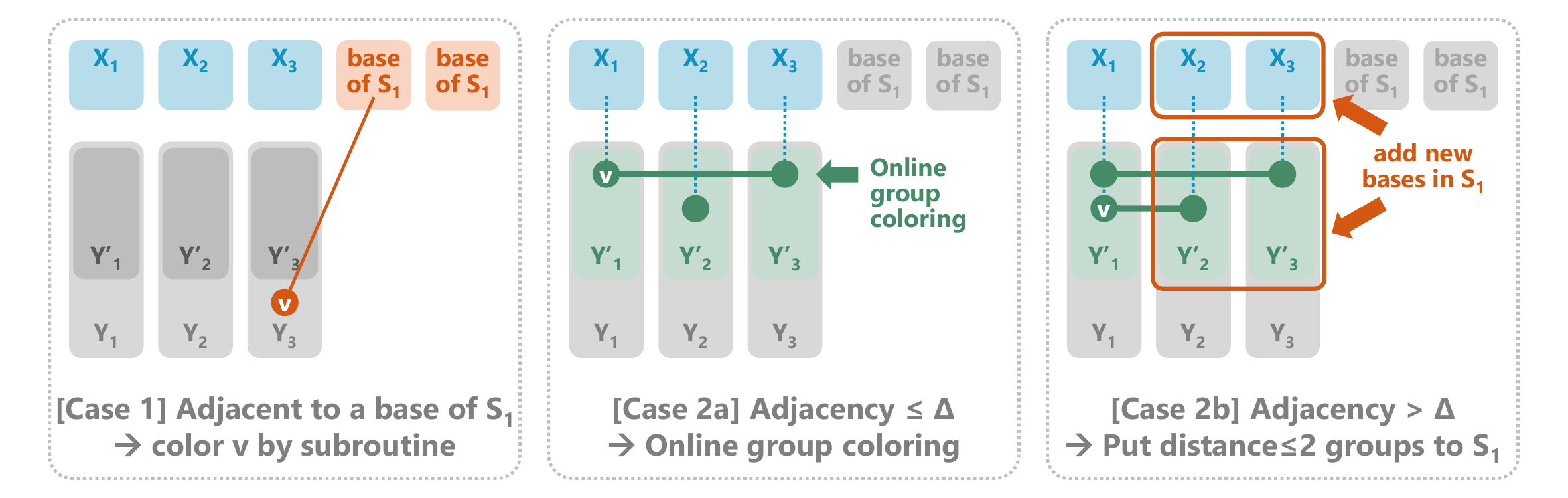}
  \caption{A sketch of the solution to Part B and the relevant variables. (a) When $v$ is adjacent to a base in $\mathcal{S}_1$, we color $v$ in the $(n^{2/5}, 24)$-subroutine and do not add $v$ to $Y'_i$. (b) Otherwise, if $|N_{H^+}(i)| \leq \Delta$, we color $v$ using our algorithm for online group coloring. (c) Otherwise, we merge the distance-$0, 1, 2$ groups and put them into $\mathcal{S}_1$ as bases.}
  \label{fig:main-algorithm}
\end{figure}

\begin{enumerate}
    \item If $v$ is adjacent to any existing base in $\mathcal{S}_1$, we color $v$ inside $\mathcal{S}_1$. Formally, we pick one such base, and add $v$ to the corresponding group for this base in $\mathcal{S}_1$ (as a coloring query). This means that we use our algorithm for the $(n^{2/5}, 24)$-subroutine to color $v$.
    \item Otherwise, after adding $v$ to $Y'_i$ and updating the graph $H^+$, we proceed as follows:
    \begin{enumerate}
        \item If $|N_{H^+}(i)| \leq \Delta$, then we put $v$ into group $i$ for online group coloring and color $v$ using the algorithm in \autoref{thm:group-coloring-1}. This is because the degree requirement for online group coloring is satisfied as the neighbor graph is a subgraph of $H^+$.
        \item Otherwise, we cannot color $v$ with online group coloring. Instead, we exploit the large degree and create large sets with even-diameter $\leq 24$. Let $D_0, D_1, D_2$ be the sets of vertices with distance $0, 1, 2$ from vertex $i$ in $H^+$. It can be shown that for $k = 0, 1, 2$, the sets $\bigcup_{j \in D_k} X_j$ and $\bigcup_{j \in D_k} Y'_j$ indeed have even-diameter $\leq 24$ (\autoref{lem:x2-diameter}). Hence, we add these $6$ sets as bases for $\mathcal{S}_1$. After that, we can color $v$ by going back to Step 1, since $v$ is adjacent to a new base $\bigcup_{j \in D_0} X_j \ (= X_i)$.
    \end{enumerate}
\end{enumerate}

The pseudocode of this algorithm is given in \autoref{alg:main-0.4}.

\begin{algorithm}[h]
    \caption{Algorithm to solve the $(n^{3/5}, 2)$-subroutine with $O(n^{2/5})$ colors. Here, we define the graph $H^+ = (\{1, \dots, r\}, E(H^+))$ by $E(H^+) := \{i-j:$ edges exist between $Y'_i$ and $Y'_j$ in $G$$\}$.}
    \label{alg:main-0.4}
    \begin{algorithmic}[1]
        \State $\mathcal{S}_1 \gets$ the $(n^{2/5}, 24)$-subroutine
        \For {each coloring query (supposing that vertex $v$ arrives in group $Y_i$)}
            \If {$v$ is adjacent to any existing base in $\mathcal{S}_1$}
                \State color $v$ inside $\mathcal{S}_1$ (by adding $v$ to $\mathcal{S}_1$ as a coloring query)
            \Else
                \State $Y'_i \gets Y'_i \cup \{v\}$, and update the graph $H^+$
                \If {$|N_{H^+}(i)| \leq \Delta \ (= 6n^{1/5})$}
                    \State color $v$ with our algorithm for online group coloring (\autoref{thm:group-coloring-1})
                \Else
                    \State $D_k \gets$ the set of vertices with distance $k$ from vertex $i$ in $H^+$, for $k = 0, 1, 2$
                    \State add $\bigcup_{j \in D_k} X_j$ and $\bigcup_{j \in D_k} Y'_j$ as bases for $\mathcal{S}_1$, for $k = 0, 1, 2$
                    \State go to lines 3--4 to color $v$ (because $v$ neighbors a new base $\bigcup_{j \in D_0} X_j = X_i$)
                \EndIf
            \EndIf
        \EndFor
    \end{algorithmic}
\end{algorithm}

\begin{remark} \label{remark:alg-future}
    We note that, after we perform Step 2-b with respect to $D_0, D_1, D_2$, all future vertices $v$ that arrive in group $Y_i \ (i \in D_0 \cup D_1 \cup D_2)$ will be processed in Step 1 (and therefore will not proceed to Step 2, which makes $Y'_i$ ``frozen''). This is because if $i \in D_k \ (k = 0, 1, 2)$, then $v$ is adjacent to $\bigcup_{j \in D_k} X_j$, which is a base already added to $\mathcal{S}_1$.
\end{remark}


\subsection{Correctness \& Analysis}

Next, we prove the correctness of this algorithm. We show that the input to the online group coloring problem and the $(n^{2/5}, 24)$-subroutine $\mathcal{S}_1$ is valid, and in the resulting coloring, no two adjacent vertices share the same color (\autoref{lem:special-valid}). After that, we analyze the number of colors used (\autoref{lem:special-part-a}).

\begin{lemma}\label{lem:x2-diameter}
    Bases added to the $(n^{2/5}, 24)$-subroutine $\mathcal{S}_1$ have even-diameter at most 24.
\end{lemma}

\begin{proof}
    First, we show that for any $k \in \{0, 1, 2\}$, in Step 2-b of the algorithm, $\bigcup_{j \in D_k} X_j$ has even-diameter $\leq 22$. That is, for any pair of vertices $v^{(s)}, v^{(e)} \in \bigcup_{j \in D_k} X_j$, there exists an even-length walk of length at most $22$.
    
    Let $X_{g^{(s)}}$ and $X_{g^{(e)}}$ be bases to which $v^{(s)}$ and $v^{(e)}$ belong, respectively. Let $g_0 \to g_1 \to \dots \to g_{2k} \ (g_0 = g^{(s)}, g_{2k} = g^{(e)})$ be a walk of length $2k$ between $g^{(s)}$ and $g^{(e)}$ in the graph $H^+$. This walk exists by the definition of $D_k$.

    For $i = 0, \dots, 2k-1$, let $e_i - s_{i+1}$ (where $e_i \in Y'_{g_i}$ and $s_{i+1} \in Y'_{g_{i+1}}$) be an edge between $Y'_{g_i}$ and $Y'_{g_{i+1}}$. This exists by the definition of the graph $H^+$. Let $s'_i$ and $e'_i$ be vertices in $X_{g_i}$ that neighbor $s_i$ and $e_i$, respectively. Also, let $s'_0 := v^{(s)}$ and $e'_{2k} := v^{(e)}$; note that $s'_0 \in X_{g_0}$ and $e'_{2k} \in X_{g_{2k}}$. Then, we obtain the following walk (see \autoref{fig:ours-proof} for visualization):
    \begin{align*}
        (v^{(s)} =) & \ [s'_0 \to \cdots \to e'_0 \to e_0] \to [s_1 \to s'_1 \to \cdots \to e'_1 \to e_1] \to \cdots \\
        & \to [s_{2k-1} \to s'_{2k-1} \to \cdots \to e'_{2k-1} \to e_{2k-1}] \to [s_{2k} \to s'_{2k} \to \cdots \to e'_{2k}] \ (= v^{(e)})
    \end{align*}
    where each $s'_i \to \cdots \to e'_i$ is an even-length walk of length at most $d^* = 2$; this exists because each $X_{g_i}$ has even-diameter $\leq 2$. Hence, the length of this walk from $v^{(s)}$ to $v^{(e)}$ is an even number at most $(d^*+3)(2k+1) - 3 \leq 5d^*+12 = 22$, since $k \leq 2$.

    Next, we prove the claim for $\bigcup_{j \in D_k} Y'_j$. Consider any $v^{(s)}, v^{(e)} \in \bigcup_{j \in D_k} Y'_j$; since both $v^{(s)}$ and $v^{(e)}$ are adjacent to $\bigcup_{j \in D_k} X_j$, there is a walk between $v^{(s)}$ and $v^{(e)}$ whose length is $2$ longer than in the case above. Therefore, $\bigcup_{j \in D_k} Y'_j$ has even-diameter $5d^* + 14 \leq 24$.
\end{proof}

\begin{figure}[t]
  \centering
  \includegraphics[width=\linewidth]{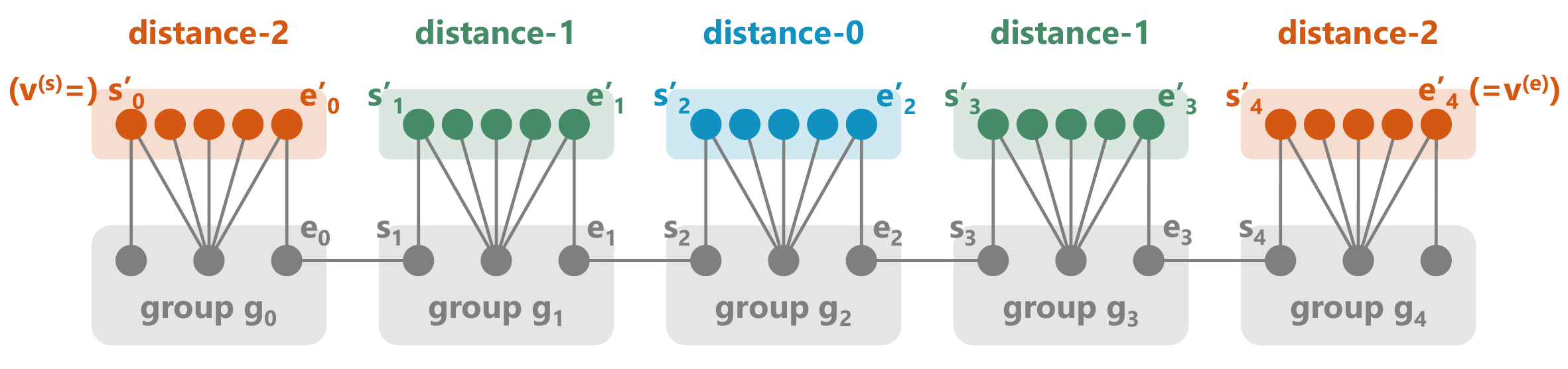}
  \caption{A sketch of the proof that $\bigcup_{j \in D_k} X_j$ has even-diameter at most 22, for $k = 2$. Indeed, the length of the walk from $v^{(s)}$ to $v^{(e)}$ is exactly 22, which is the worst case.}
  \label{fig:ours-proof}
\end{figure}

\begin{lemma} \label{lem:x2-dist3}
    Let $Z = [z_1, z_2, \dots]$ be the list, where $z_i$ is the (sole) element of $D_0$ in the $i$-th occurrence of Step 2-b. Then, any pair of elements in $Z$ are at least distance 3 apart in the final graph $H^+$.
\end{lemma}

\begin{proof}
    Consider $D_0, D_1, D_2$ in the $i$-th occurrence of Step 2-b (i.e., at the time of $z_i$). Then, $D_0, D_1, D_2$ are equal to the set of vertices with distance $0, 1, 2$ from $z_i$ in the \emph{final graph} $H^+$. In other words, the set of vertices with distance $0, 1, 2$ from $z_i$ remains unchanged.

    Suppose, for contradiction, that this is not the case. Then, an edge incident to $D_k \ (k = 0, 1, 2)$ is newly added to $H^+$, after the occurrence of $D_k$ (i.e., after the time of $z_i$). Let $v'$ be a vertex that caused this edge addition. Then, $v'$ proceeds to Step 2, and by \autoref{remark:alg-future}, $v'$ is not in groups of $D_0 \cup D_1 \cup D_2$. Thus, $v'$ is adjacent to some $Y'_j \ (j \in D_k)$. This implies that $v'$ is also adjacent to $\bigcup_{j \in D_k} Y'_j$, which is a base already added to $\mathcal{S}_1$. Thus, $v'$ should have been proceeded in Step 1, a contradiction.

    For each $j = i+1, i+2, \dots$, we have $z_j \notin D_0 \cup D_1 \cup D_2$ by \autoref{remark:alg-future}. Hence, with the observation above, $z_j$ is at least distance $3$ apart from $z_i$ in the final graph $H^+$ as well.
\end{proof}

\begin{lemma} \label{lem:x2-size}
    At most $n^{2/5}$ bases are added to the $(n^{2/5}, 24)$-subroutine $\mathcal{S}_1$.
\end{lemma}

\begin{proof}
    Let $Z = [z_1, z_2, \dots]$ be the list defined in \autoref{lem:x2-dist3}. Then, by \autoref{lem:x2-dist3}, the sets $N_{H^+}(z_1), N_{H^+}(z_2), \dots$ are disjoint for the final graph $H^+$. Moreover, we have $|N_{H^+}(z_i)| > 6n^{1/5} \ (= \Delta)$ because they were not processed in Step 2-a. Since $H^+$ has at most $r^* = n^{3/5}$ vertices, we have $|Z| \leq \frac{n^{3/5}}{6n^{1/5}} = \frac{1}{6} n^{2/5}$. Thus, at most $6|Z| \leq n^{2/5}$ bases are added to $\mathcal{S}_1$.
\end{proof}

\begin{lemma} \label{lem:special-valid}
    For graphs with odd girth at least 29, the inputs for (1) online group coloring and (2) the $(n^{2/5}, 24)$-subroutine $\mathcal{S}_1$ are valid.
\end{lemma}

\begin{proof}
    For (1) online group coloring, the degree conditions are met due to Step 2-a. Also, there are no edges connecting vertices in the same group, because there are no edges inside each $Y_i$. Otherwise, given that each base $X_i$ has even-diameter at most $d^* = 2$, there would exist an odd closed walk of length at most $d^* + 3 = 5$, by the same argument as in the proof of \autoref{lem:special-part-b}. This contradicts that the odd girth is at least 29.

    For (2) the $(n^{2/5}, 24)$-subroutine $\mathcal{S}_1$, the added bases have even-diameter $\leq 24$, due to \autoref{lem:x2-diameter}. When the vertices are added to $\mathcal{S}_1$, they are adjacent to one of its bases, due to Step 1 of the algorithm. At most $n^{2/5}$ bases are added, due to \autoref{lem:x2-size}.
\end{proof}

\begin{lemma} \label{lem:special-part-a}
    For graphs with odd girth at least 29, \autoref{alg:main-0.4} solves the $(n^{3/5}, 2)$-subroutine deterministically using $O(n^{2/5})$ colors.
\end{lemma}

\begin{proof}
    \autoref{lem:special-valid} shows that the algorithm functions properly. Therefore, we use at most $\Delta^2 + 2 = 36n^{2/5} + 2$ colors for online group coloring due to \autoref{thm:group-coloring-1} and $\Delta = 6n^{1/5}$, and we use at most $n^{2/5}$ colors in the $(n^{2/5}, 24)$-subroutine $\mathcal{S}_1$ due to \autoref{lem:special-part-b}. In total, we use $O(n^{2/5})$ colors.
\end{proof}

\begin{theorem}\label{thm:main-1}
    For graphs of odd girth at least 29, there is a deterministic online coloring algorithm that uses $O(n^{2/5})$ colors.
\end{theorem}

\begin{proof}
    By combining \autoref{lem:special-base-case} and \autoref{lem:special-part-a}, we obtain a deterministic online coloring algorithm that uses $O(n^{2/5})$ colors for graphs with odd girth at least 29.
\end{proof}

\section{The Algorithm with $O(n^\varepsilon)$ colors} \label{sec:general}

In this section, we generalize the idea from \autoref{sec:special} to multiple layers. We reduce the exponent --- from $\frac{1}{2}$ in Kierstead's algorithm and $\frac{2}{5}$ in \autoref{sec:special} --- to an arbitrarily small positive value for graphs with sufficiently large odd girth. In other words, we prove \autoref{thm:generalize-1}.

\subsection{The Generalized Lemmas}

We revisit three key steps in constructing our $O(n^{2/5})$-color algorithm: \autoref{lem:special-base-case} (Part A), \autoref{lem:special-part-a} (Part B), and \autoref{lem:special-part-b} (Part C). Below, we generalize these lemmas from the specific $O(n^{2/5})$ setting to one with general parameters.

\begin{lemma} \label{lem:general-a}
    Let $c \geq 1$ be an integer parameter. The online coloring problem with $n$ vertices can be reduced to the $(\frac{n}{c}, 2)$-subroutine, deterministically, using at most $c$ extra colors.
\end{lemma}

\begin{proof}
    This follows from the same argument as \autoref{lem:special-base-case}, where we use $c$ colors for First-Fit (instead of $\lceil n^{2/5} \rceil$ colors). Then, the number of bases created is at most $\frac{n}{c}$.
\end{proof}

\begin{lemma} \label{lem:general-b}
    Let $\Delta \geq 1$ be a real parameter. For graphs with odd girth at least $d^* + 5$, the $(r^*, d^*)$-subroutine can be reduced to the $(\frac{6r^*}{\Delta}, 5d^*+14)$-subroutine, deterministically, using at most $\Delta^2 + 2$ extra colors.
\end{lemma}

\begin{proof}
    This follows from the same argument as \autoref{lem:special-part-a}, where we use $\Delta$ as the parameter for online group coloring. The even-diameter bound $5d^* + 14$ comes from the proof of \autoref{lem:x2-diameter}. The maximum number of bases created, $\frac{6r^*}{\Delta}$, comes from the proof of \autoref{lem:x2-size}. The number of extra colors used, $\Delta^2 + 2$, comes from \autoref{thm:group-coloring-1}. We need the odd girth requirement to guarantee that there is no edge inside the same group in online group coloring; see the proof of \autoref{lem:special-valid}.
\end{proof}

\begin{lemma} \label{lem:general-c}
    For graphs with odd girth at least $d^* + 5$, the $(r^*, d^*)$-subroutine can be solved deterministically using at most $r^*$ colors.
\end{lemma}

\begin{proof}
    This follows from the same argument as \autoref{lem:special-part-b}.
\end{proof}

\subsection{The Generalized Algorithm}

Reflecting on the previous algorithms, Kierstead's algorithm for graphs with odd girth $\geq 7$ reduced the online coloring problem via \autoref{lem:general-a} and \autoref{lem:general-c}, achieving $O(n^{1/2})$ colors. Our algorithm for graphs with odd girth $\geq 29$ reduced the online coloring problem via \autoref{lem:general-a}, \autoref{lem:general-b}, and \autoref{lem:general-c}, achieving $O(n^{2/5})$ colors.

The good news is that the reduction via \autoref{lem:general-b} can be applied multiple times --- this is key to constructing a further improved algorithm. More specifically, our strategy is to apply this reduction through $k$ layers to achieve $O(n^{\frac{2}{k+4}})$ colors.

\begin{theorem} \label{thm:generalize-2}
    Let $k \geq 0$ be an integer. For graphs of odd girth at least $\frac{11}{2} \cdot 5^k + \frac{3}{2}$, there exists a deterministic online coloring algorithm that uses $O(k \cdot n^{\frac{2}{k+4}})$ colors.
\end{theorem}

\begin{proof}
    First, we reduce the online coloring problem to the $(n^{\frac{k+2}{k+4}}, 2)$-subroutine using \autoref{lem:general-a} with parameter $c := \lceil n^{\frac{2}{k+4}} \rceil$. This uses at most $\lceil n^{\frac{2}{k+4}} \rceil$ colors.
    
    Next, we apply the reduction from \autoref{lem:general-b}, $k$ times, each with parameter $\Delta := 6n^{\frac{1}{k+4}}$. In the $\ell$-th reduction $(\ell = 1, \dots, k)$, we reduce the $(n^{\frac{k-\ell+3}{k+4}}, a_{\ell-1})$-subroutine to the $(n^{\frac{k-\ell+2}{k+4}}, a_{\ell})$-subroutine, where $a_{\ell}$ is defined by $a_0 = 2$ and $a_{\ell} = 5 a_{\ell-1} + 14 \ (\ell = 1, 2, \dots)$. This uses at most $k \cdot (\Delta^2 + 2) = k \cdot (36n^{\frac{2}{k+4}} + 2)$ colors in total.

    Finally, we solve the $(n^{\frac{2}{k+4}}, a_k)$-subroutine by \autoref{lem:general-c}, using $n^{\frac{2}{k+4}}$ colors.

    Overall, the number of colors used is $O(k \cdot n^{\frac{2}{k+4}})$. The odd girth requirement is needed to use \autoref{lem:general-b} and \autoref{lem:general-c}, but odd girth $\geq a_k + 5$ is sufficient for all steps of the algorithm, which equals $\frac{11}{2} \cdot 5^k + \frac{3}{2}$. This follows from the closed form $a_k = \frac{11}{2} \cdot 5^k - \frac{7}{2}$.
\end{proof}

In conclusion, with \autoref{thm:generalize-2} in hand, we claim our main theorem (\autoref{thm:generalize-1}): for every $\varepsilon > 0$, there exists a constant $g' \in \{3, 5, 7, \dots\}$ such that graphs with odd girth at least $g'$ can be deterministically colored online using $O(n^{\varepsilon})$ colors. \autoref{tab:colors-girth} shows how the number of colors used depends on the odd girth. The algorithm clearly runs in polynomial time.

\subsection{Performance with Non-Constant Odd Girth}

When the number of layers is $k = \Theta(\log n)$, it follows from \autoref{thm:generalize-2} that the number of colors used in the algorithm is just $O(\log n)$. Hence, we obtain the following corollary.

\begin{corollary} \label{cor:generalize-3}
    Let $c$ be a constant such that $0 < c < 1$. For graphs with odd girth $\Omega(n^c)$, there exists a deterministic online coloring algorithm that uses $O(\log n)$ colors.
\end{corollary}

\begin{proof}
    Applying \autoref{thm:generalize-2} with $k = (c - \varepsilon) \log_5 n$ for some $\varepsilon \ (0 < \varepsilon < c)$, we obtain a deterministic online algorithm that colors graphs with odd girth at least $\frac{11}{2} \cdot n^{c-\varepsilon} + \frac{3}{2}$, using only $O(\log n)$ colors. Therefore, \autoref{cor:generalize-3} holds.
\end{proof}

This substantially improves the result of $O(n^{(1-c)/2} \sqrt{\log n})$ colors by Nagy-Gy\"{o}rgy \cite{Nag10} and matches the bipartite graph bound of $\Theta(\log n)$ colors due to Lov\'{a}sz, Saks, Trotter \cite{LST89}.

\section{Conclusion}

In this paper, we showed a deterministic online coloring algorithm that uses $O(k \cdot n^{\frac{2}{k+4}})$ colors for graphs with odd girth $g \geq \frac{11}{2} \cdot 5^k + \frac{3}{2}$, for every integer $k \geq 0$. This gives the first improvement in 28 years for online coloring of graphs with large constant odd girth, improving upon the $O(n^{1/2})$-color algorithm of Kierstead (1998) \cite{Kie98}.

It is known that any online algorithm requires $\Omega(n^{\frac{1}{g-3}} / \log n)$ colors for graphs with general odd girth $g \geq 5$, as implied by the chromatic number bound of Krivelevich (1995) \cite{Kri95}. Therefore, for every constant $g$, there exists an instance requiring $\Omega(n^{\varepsilon})$ colors for some $\varepsilon > 0$. In this sense, since we achieve $O(n^{\varepsilon})$ colors for sufficiently large odd girth, our algorithm is in the same ballpark as this lower bound.

However, in terms of the speed of convergence, our exponent is approximately $\frac{2}{\log_5 g}$, which is much slower than the exponent $\frac{1}{g-3}$ in the lower bound. Therefore, it is natural to ask whether the following conjecture holds.

\begin{conjecture}
    For some constant $c \geq 1$, there exists an online coloring algorithm that achieves $O(n^{c/g})$ colors for graphs with constant odd girth $g$.
\end{conjecture}

\section*{Acknowledgments}

We are grateful to our professor Ken-ichi Kawarabayashi for helpful discussions and comments.

\printbibliography

\end{document}